\documentclass[reqno]{amsart}

\usepackage{graphicx,epsfig}
\usepackage{multirow}
\usepackage{algorithm}
\usepackage{algpseudocode}

\usepackage{amsmath,amssymb,latexsym, amsfonts, amscd, amsthm, xy}
\input{xy}
\xyoption{all}
\usepackage{diagbox}

\input{xy}
\xyoption{all}

\numberwithin{equation}{section}

\usepackage[usenames]{color}

\usepackage{parskip}

\makeindex
=651pt \headheight = 12pt \headsep = 0pt

\theoremstyle{plain}
\newtheorem{theorem}{Theorem}[section]

\newtheorem{proposition}[theorem]{Proposition}
\newtheorem{corollary}[theorem]{Corollary}

\newtheorem{definition}[theorem]{Definition}
\newtheorem{remark}{Remark}[section]
\newtheorem{example}[theorem]{Example}

\def\R{\mathbb{R}}
\def\SPD{\text{SPD}}

\def\sym{\text{Sym}}

\def\t{\textrm}
\def\bb{\mathbb}

\def\R{\mathbb{R}}
\def\bb{\mathbb}

\usepackage{natbib}

\title{Extrinsic Gaussian processes  for  regression and classification on manifolds}

\author{Lizhen Lin, Mu Niu, Pokman  Cheung,  and David B. Dunson}

\address{Department of Applied and Computational Mathematics and Statistics, The University of Notre Dame,  USA}
	      \email{lizhen.lin@nd.edu}
	      
    \address{School of Mathematics and Statistics,  University of Glasgow, United Kingdom }

\address{Department of Mathematics,  Hong Kong University of Science and Technology, China}

              \address{Department of Statistical Science, Duke University, USA}

\begin{document}






\begin{abstract}

Gaussian processes (GPs) are very widely used for modeling of unknown functions or surfaces in applications ranging from regression to classification to spatial processes.  Although there is an increasingly vast literature on applications, methods, theory and algorithms related to GPs, the overwhelming majority of this literature focuses on the case in which the input domain corresponds to a Euclidean space.  However, particularly in recent years with the increasing collection of complex data, it is commonly the case that the input domain does not have such a simple form.  For example, it is common for the inputs to be restricted to a non-Euclidean manifold, a case which forms the motivation for this article.  In particular, we propose a general extrinsic framework for GP modeling on manifolds, which relies on embedding of the manifold into a Euclidean space and then constructing extrinsic kernels for GPs on their images.  These extrinsic Gaussian processes (eGPs) are used as prior distributions for unknown functions in Bayesian inferences. Our approach is simple and general, and we show that the eGPs inherit fine theoretical properties from GP models in Euclidean spaces. We consider applications of our models to regression and classification problems with predictors lying in a large class of manifolds, including spheres, planar shape spaces, a space of positive definite matrices, and  Grassmannians.  Our models can be readily used by practitioners in biological sciences for various regression and classification problems, such as disease diagnosis  or detection. Our work is also likely to have impact in spatial statistics when spatial locations are on the sphere or other geometric spaces.


\end{abstract}

\maketitle

\section{Introduction}

Over the past few decades, Gaussian process (GP) models have emerged as  very powerful tools in many problems of statistics and machine learning.  In particular,  GP  models have been widely used in  regression and classification, in which a  Gaussian process is  used as the  prior distribution for the regression function or the latent function of a classification map.
 GP models are particularly appealing in their ability to accurately quantify uncertainty in estimation and prediction.  
 \cite{rasmussen_gaussian_2005} provide  an overview on GPs in machine learning. \cite{vandervaart2008, vandervaart2009}  develop  theoretical guarantees of GP models in terms of support and posterior asymptotic theory.  However,  few attempts have been made in developing applicable GP models for \emph{regression and classifications on manifolds} except for some very special cases,  such as the 2-dimensional sphere \citep{Hitczenko2012211, guinness2016isotropic}.

 One of the paramount challenges in developing GP  models on manifolds  is constructing  \emph{valid covariance kernels}.   \cite{casmillo} develop an elegant framework for intrinsic GP models on Riemannian manifolds by rescaling  solutions of heat equations,  but the constructed intrinsic kernels are often impractical to implement.   We provide a general and simple solution by first embedding manifolds into Euclidean spaces via \emph{equivariant embeddings}, which are embeddings that preserve a great deal of the geometry of the manifolds, and then \emph{constructing extrinsic kernels} on the image manifold. We refer to the resulting GPs as \emph{extrinsic GPs (eGPs)}.  eGPs are shown to inherit  appealing properties of GPs defined on Euclidean spaces, and they adapt to the dimension of the manifolds instead of the dimension of  the ambient space where the manifolds are embedded onto.  Another appealing feature of eGPs is their ease of implementation for inference.

One of the motivations for developing GP models on manifolds is the ubiquity of modern data that are represented in various non-conventional forms. In neuroimaging, the diffusion matrices in diffusion tensor imaging (DTI) are $3\times3$ \emph{positive definite matrices} \citep{dti-ref}. In engineering and machine learning, pictures or images are often  preprocessed or reduced to a collection of \emph{subspaces} \citep{subspacepaper, facialsubpace}.
 In machine vision and medical diagnostics, a digital image can also be represented by a set of $k$-landmarks, the collection of which form  \emph{landmark-based shape spaces} \citep{kendall84}.  Other  common examples include  \emph{orthonormal frames} \citep{vecdata}, \emph{surfaces}, \emph{curves}, and \emph{networks}.  Most of the above examples can be described as \emph{manifolds}, which are locally Euclidean spaces with smooth structures.

 There are growing needs and  practical motivations for studying regression and classification with predictors on known manifolds. For instance, in medical imaging, a common goal  is to  reliably predict disease status using DTI data or landmark-based digital images. This can be viewed as a classification problem with manifold-valued inputs or predictors.   One example is diagnosis of Attention Deficit Hyperactivity Disorder (ADHD) in children based on DTI. There are also many applications in which it is of interest to relate manifold-valued predictors to quantitative traits. One such case is the study of how intelligence quotient relates to the shape contours of certain brain areas (such as the Hippocampus \citep{hippo}). The shape can be represented by a set of landmarks on the boundary of the contours, the collection of which form a shape manifold. Without valid models and appropriate inferential methods for regression and classification on manifolds, making accurate inferences and predictions in the above applications and related settings will remain difficult.


There is already a rich literature on statistical inference for manifold-valued data consisting of i.i.d measurements. Much of this literature focuses on inference on the location and spread of manifold-valued data  \citep{rabi2003, Bhattacharya2005, linclt}. Some model based methods  have also been proposed  \citep{Bhattacharya01122010, sinica-paper, Pelletier2005297}.  However, regression or classification problems with predictors on manifolds have received much less attention. \cite{david-regression} proposed a framework for regression and classification on manifolds by  modeling the joint distribution  of covariate and response variables $(x,y)$  using a Dirichlet process mixture of product kernels.  This joint model induces a nonparametric model for the conditional distribution of $y$ given $x$ with which one can infer the regression/classification function. However, the practical performance of these models is often unsatisfactory  as the cluster allocations are driven too much by the marginal distribution of $x$,  a nuisance parameter.

Our work focuses on regression and classification on \emph{known manifolds}.  There is, however,  an important  line of work in manifold learning, where the predictors concentrate around some \emph{unknown lower-dimensional manifold} but are observed in an often higher-dimensional ambient space. The lower-dimensional geometry is often learnt first via dimension reduction tools, based on which a regression model is built (see, e.g., \cite{yen13}).  An  interesting exception is due to \cite{yang2016} in which they show that by imposing a Gaussian process prior on the regression function with a covariance kernel defined directly on the ambient  space,  the  posterior distribution yields a posterior contraction rate depending on the intrinsic dimension of the manifold.  They assume that the unknown lower-dimensional space where the predictors center around are a class of \emph{submanifolds} of Euclidean space.  Many interesting manifolds do not naturally arise as sub-manifolds; in particular, those given as quotient manifolds;  projective shape spaces,  planar shapes, 3-$D$ shapes, affine shapes and many other  manifolds arising as quotient spaces of spheres.  Our framework first embeds the manifold onto the Euclidean space  via some often non-trivial embeddings and  then defines eGPs on the image of the manifolds (including submanifolds as  special cases with the embedding given by the identity map).


The paper is organized as follows.  Sections 2 introduces eGP models.  In section 3, we illustrate the broad utility of eGP models by applying them to a large class of regression/classification problems with predictors lying on various manifolds.   Section 4 is devoted to studying the properties of eGP models in terms of mean squared differentiability and posterior contraction rates. Our paper ends with a discussion.
\section{Regression and classification on manifolds}

Let $M$ be a smooth manifold where the predictors lie.   Given data $(x_i,y_i)$ with $x_i\in M$ and $y_i\in \R$ ($i=1,\ldots, n$), assume the following  regression model
\begin{align}
\label{eq-model1}
y_i=F(x_i)+\epsilon_i
\end{align}
where $F: M\rightarrow \R$ is the regression function on $M$. Here $\epsilon_i$'s are some independent errors which determine  the likelihood of the regression model. The goal is to develop statistical models for inference on the regression function $F(x)$. If $y$ is categorical or binary (0 or 1), then $F(x)=E(y\mid x)$ is called a \emph{classification map}.


We focus on Bayesian inference on $F$. Let $\Pi(F)$ be a prior distribution for $F$, which updates with the data to produce a posterior distribution, based on which inference is carried out.
 We  denote the posterior distribution by  $\Pi(F|D)$,  where $D=\{(x_1,y_1),\ldots, (x_n,y_n)\}$ is the data.  A Gaussian process (GP), which can be viewed as a probability distribution on the space of functions, is one of the most popular candidates for a nonparametric prior for the regression function.  The popularity of GP is due to its simple representation, tractability, flexibility for modeling and  appealing theoretical properties.
We proceed to   propose a general extrinsic framework for constructing \emph{GPs on manifolds}.




The usual definition of a GP in a Euclidean space generalizes to a manifold $M$. A stochastic process $w(x)$ indexed by $x\in M$ is a \emph{Gaussian process on $M$} if its evaluation at any finite number of points on $M$ follows a multivariate Gaussian distribution. Specifically, we say $w(x)$ is a GP with mean function $\mu(x)$ and covariance kernel $K(\cdot,\cdot)$ if for any $x_1,\ldots,x_n\in M$,
\begin{align*}
&(w(x_1),\ldots, w(x_n))\;\sim\; N\left((\mu(x_1),\ldots, \mu(x_n)), \Sigma\right),\\
&\qquad\textrm{where }\Sigma_{ij}=cov\left(w(x_i), w(x_j)\right)=K(x_i,x_j).
\end{align*}
Notice that $K:M\times M\rightarrow \R$ is a \emph{positive semi-definite kernel} on $M$. Namely, for any points $x_1,\ldots,x_n$ on $M$ and real numbers $a_1,\ldots,a_n$,
\begin{align}
\label{eq-kernel}
\sum_{i=1}^n\sum_{j=1}^n a_i a_j K(x_i,x_j)\geq 0.
\end{align}
The fundamental difficulty in imposing a GP prior on a manifold stems from the highly challenging task of constructing a valid covariance kernel $K(\cdot,\cdot)$.
Below we describe a simple recipe for constructing valid covariance kernels using an extrinsic approach.


Let $J: M\rightarrow\bb{R}^D$ be an embedding of $M$ into some higher dimensional Euclidean space $\bb{R}^D$ ($D\geq\dim M$) and denote the image of the embedding as $\widetilde{M}=J(M)$. By definition of an embedding, $J$ is a smooth map such that its differential at each point $x\in M$ is an injective map (from the tangent space of $M$ at $x$ to the tangent space of $\bb{R}^D$ at $J(x)$), and $J$ is a homeomorphism between $M$ and its image $\widetilde{M}$.
Given a positive semi-definite kernel $\widetilde{K}$ on $\bb{R}^D$, we can then define a positive semi-definite kernel  (and hence the covariance kernel of a GP) on $M$ by
\begin{align}
\label{eq-extrinsiker}
K_{ext}(x_1,x_2)=\widetilde{K}(J(x_1),J(x_2)).
\end{align}
Indeed, $K_{ext}$ satisfies condition \eqref{eq-kernel} on $M$ because $\widetilde{K}$ satisfies the same condition on $\bb{R}^D$, hence in particular on $\widetilde{M}\subset\bb{R}^D$. We call the Gaussian process with the covariance kernel $K_{ext}(\cdot,\cdot)$ defined above an \emph{extrinsic Gaussian process (eGP)}.

\begin{remark}
Let $||\cdot||$ be the Euclidean norm.
We define the \emph{extrinsic distance} on the manifold $M$ as
\begin{align}
\label{eq-extrinsicdis}
\rho(x_1,x_2)=\| J(x_1)-J(x_2)\|.
\end{align}
One can immediately generalize the popular squared exponential kernel in Euclidean spaces to manifolds by letting
\begin{align}
\label{eq-extrinsiker2}
K_{ext}(x_1,x_2)=\alpha\exp(-\beta \rho^2(x_1,x_2)),
\end{align}
where $\rho(x_1,x_2)$ is the extrinsic distance given in \eqref{eq-extrinsicdis}.
One can also generalize the class of Mat\'ern covariance kernels to manifolds by letting
\begin{align}
\label{eq-extrinmearn}
K_{ext}(x_1,x_2)=\sigma^2\frac{1}{\Gamma(\nu)2^{\nu-1}}\left(\frac{\sqrt{2\nu}\rho(x_1,x_2)}{\kappa}\right)^{\nu}K_{\nu}\left(\frac{\sqrt{2\nu}\rho(x_1,x_2)}{\kappa}\right),
\end{align}
where $\Gamma(\nu)$ is the Gamma function,  $K_{\nu}$ is the modified Bessel function of the second kind, and $\kappa$ and $\nu$ are non-negative parameters of the covariance. Mat\'ern covariance kernels are often used in spatial statistics with which one can easily control the smoothness of the sample paths with parameter $\nu$.
The following is clear.
\end{remark}

\begin{proposition}
The kernels given in \eqref{eq-extrinsiker2} and \eqref{eq-extrinmearn} are positive semi-definite kernels on $M$.
\end{proposition}

\begin{figure}[h]
\begin{center}
\includegraphics[width=11cm]{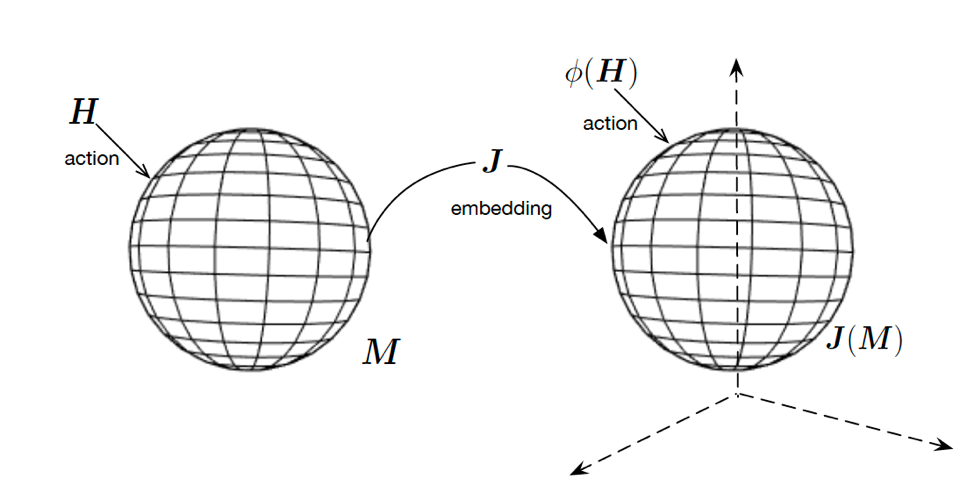}
\end{center}
\caption{An simple illustration of equivariant embeddings}
\label{fig-embed}
\end{figure}

\begin{remark}

The embedding $J$ is never unique. It is desirable to have  an embedding that preserves as much geometry as possible.  An \emph{equivariant embedding}  is one type of embedding that preserves a substantial amount of geometry.  Figure \ref{fig-embed} provides a visual illustration.
 Suppose $M$ admits an action of a (usually `large') Lie group $H$. Then we say that $J$ is an equivariant embedding if we can find a Lie group homomorphism $\phi: H\rightarrow GL(D, \mathbb R)$ from $H$ to the general linear group $GL(D, \mathbb R)$ of degree $D$ acting on $\widetilde M$ such that
\begin{align*}
J(hp)=\phi(h)J(p)
\end{align*}
for any $h\in H$ and $p\in M$.  The definition seems technical, however, the intuition is clear: if a large group $H$ acts on  the manifolds such as by rotation before embedding, such an action can be preserved via $\phi$  on the image $\widetilde M$. Therefore, the embedding is geometry-preserving in this sense. 
\end{remark}
%

\begin{remark}

The extrinsic method described above has some advantages over using intrinsically defined covariance kernels. In particular, intrinsic kernels are difficult to construct in general.  For example, the squared exponential kernel  $\alpha\exp(-\beta \rho_g^2(x_1,x_2))$ with $\rho_g$ given by the geodesic or intrinsic distance is in general not a valid kernel. Explicit examples have been found for very special manifolds only, such as spheres. At the same time, simulation tests have shown that there is no significant difference in statistical performance between certain extrinsic and intrinsic models, at least for the example of spheres. However, intrinsic methods are often much more computationally complex and expensive.

\end{remark}

With a valid covariance kernel on $M$, one can specify an eGP as a  prior  $\Pi(F)$  and carry out inference in a Bayesian framework.
Given the regression model in \eqref{eq-model1}, we assume that $\epsilon_i\sim N(0,\sigma^2)$, where the parameter $\sigma^2$ has a prior distribution $\pi_{\sigma^2}$ such as the inverse gamma distribution.  The prior distribution for the regression function  $\Pi (F)$ will be given by the eGP with the covariance kernel in \eqref{eq-extrinsiker}.
The posterior distribution is given by
\begin{align}
\Pi\left(U\mid (x_1,y_1),\ldots, (x_n,y_n)\right)=\dfrac{\int_{U} \prod_{i=1}^n N(y_i;F(x_i),\sigma^2)\pi_{\sigma^2}\Pi(dF)}{\int \prod_{i=1}^n N(y_i;F(x_i),\sigma^2)\pi_{\sigma^2}\Pi(dF)}
\end{align}
where $U$ is a measurable set in the product space $\mathcal{M}\times (0,\infty)$ with $\mathcal{M}$ denoting the space  of all $ M \to \R$   regression functions.

Another important class of problems are classification problems, in which one is generally interested in predicting a categorical (e.g., binary as a special case) outcome given the predictors. Denote the responses or outcomes as 1 or 0 for the binary case, and let $F(x)$ be the probability of observing 1 at predictor level $x$. One can impose a prior distribution on $F$ by imposing an eGP on a latent process $w(x)$,  such that $F(x) = L(w(x))$ and $L$ is a fixed link function - for example the probit or logistic link. Properties of $F(x)$ can be derived from those for $w(x)$ as $L$ provides a smooth one-to-one monotone transformation of $w(x)$ into $L(x)$. Extensions to categorical outcomes beyond binary are straightforward.

\section{Examples}

To  illustrate the broad utility of eGP models, we consider a large class of examples with predictors lying on manifolds including spheres,  planar shapes, positive definite matrices, and Grassmannians.  All details of the embeddings are provided for constructing the extrinsic kernels for eGPs.  Embedding manifolds into Euclidean spaces or other manifolds has been  applied in different settings. In \cite{brian2014}, for example, the manifold of the parameters of a statistical model is embedded into a big sphere, while  \cite{ext2015} embeds the response manifold of a regression model into a Euclidean space for inference.  In section 3.1, a simulation study is carried out to compare the performances of an eGP model with that of  an intrinsic one in a regression model with predictors on a sphere.  In section 3.2, an eGP model is applied to classify gender of gorillas based on skull images. In this case, the predictor space is the 2-$d$  landmark-based shape space, i.e., the planar shape. In Section 3.3, we consider a classification problem whose predictors are positive definite matrices; this problem  has  important applications in neuro-imaging.  We apply the eGP model to an HIV study in identifying the most sensitive sites for disease detection or diagnostics. Lastly in section 3.4, we apply our eGP model to a regression problem with predictors lying on a Grassmannian manifold in a simulation study.


\subsection{Spheres}
Modeling on the sphere has received particular attention  due to  applications in spatial statistics; for example, global models for climate or satellite data \citep{jun2008, chunfeng}.  We consider eGP models for regression  with the predictors lying on a  sphere $S^d$.   The model is illustrated with predictors on $S^2$.  Note that for the particular case of spheres, there is a somewhat extensive literature studying valid positive-definite functions or covariance functions on the spheres for various purposes (see. e.g., \cite{gneiting2013} and \cite{Du2013}). 

To construct a valid extrinsic covariance kernel on $S^d$, first note that $S^d$ is a submanifold of $\R^{d+1}$, so that the inclusion map $J$ serves as a natural embedding of $S^d$ into $\mathbb R^{d+1}$.  It is easy to check that $J$ is an equivariant  embedding with  respect to the Lie group $H=SO(d+1)$, the group of $d+1$ by $d+1$ special orthogonal matrices. Intuitively speaking, this embedding preserves a lot of  symmetries of the sphere.


One can adopt the  extrinsic squared exponential kernel (\ref{eq-extrinsiker}) on $S^d$ for an eGP model, with
\begin{align}
K_{\t{ext}}(x,x')&\nonumber=\alpha\exp\left(-\beta\|J(x)-J(x')\|^2\right)=\alpha\exp\left(-\beta\|x-x'\|^2\right).
\end{align}

We now consider a simulation study in which the performance of an eGP model is  compared with that of a  GP model using an intrinsic kernel. Intrinsic kernels that are computation  friendly are only available for some special cases such as $S^1$ and $S^2$.  We compare our extrinsic  model to  a GP model  with the following intrinsic kernel. Letting  $d(x,x')=2\arcsin\big(\frac{1}{2}\|x-x'\|\big)$, define
\begin{align}
K_{\t{int}}(x,x')=\alpha\exp\big(-\beta d(x,x')\big),
\end{align}
which is a valid  covariance kernel on a sphere (e.g, see section 3 of  \cite{chunfeng}).

 Data are simulated from the regression model,
\begin{align}
y= F(x_1,x_2,x_3) +\epsilon
\end{align}
where $x$ is a point on the unit sphere,  $x_{1:3}$ are the coordinates of $x$ in the three dimensional Euclidean space, the true regression function $F$ is taken to be  the sum of $x_{1:3}$ and $\epsilon$ is a zero mean Gaussian noise term.   We apply a GP model with covariance kernels $K_{int}$ and $K_{ext}$.
Since the kernel parameters ($\theta=\{\alpha,\beta\}$) are correlated \citep{Rasmussen2004}, standard Markov Chain Monte Carlo (MCMC) sampling traverses the parameter space slowly. Instead, we use Hamiltonian Monte Carlo (HMC) for  inference of kernel parameters which improves efficiency by producing relatively 
distant proposals that are accepted with high probability  \citep{duane1987}.  Here are some details on the priors and the HMC chains:  both the length-scale and magnitude hyperparameters of the covariance kernels of the eGP are given gamma(10,10) priors;  $\pi_{\sigma^2}$ is given by gamma(1,10); the number of Monte Carlo iterations is 10,000 with a burn in of  1,000;  The results are not sensitive to different parameter values of the gamma distributions.

%
%
%

Two kernels are tested using $100$ samples with signal-to-noise ratio $26db$.  The true function is plotted in red and the estimate is plotted in blue in Figures \ref{fig1}.  The horizontal axis is the Euclidean coordinate $x_1$ and the vertical axis is the functional output.
The eGP model appears to produce an estimate that is closer to the true function compare to that from the intrinsic model.  Indeed, the eGP model using the kernel $K_{ext}$ yields a smaller root mean square error, which is $0.063$ compared to $0.3727$ for the intrinsic model. One of the potential reasons for  superior performance of eGP over the intrinsic model is  non-differentiability of the intrinsic distance hence  intrinsic kernel.  This non-differentiability can lead to non-smoothness of the Gaussian process (see section 4.1 for more details) thus impacting inference results.

\begin{figure}[!h]
\centering
 \includegraphics[width=0.4\linewidth]{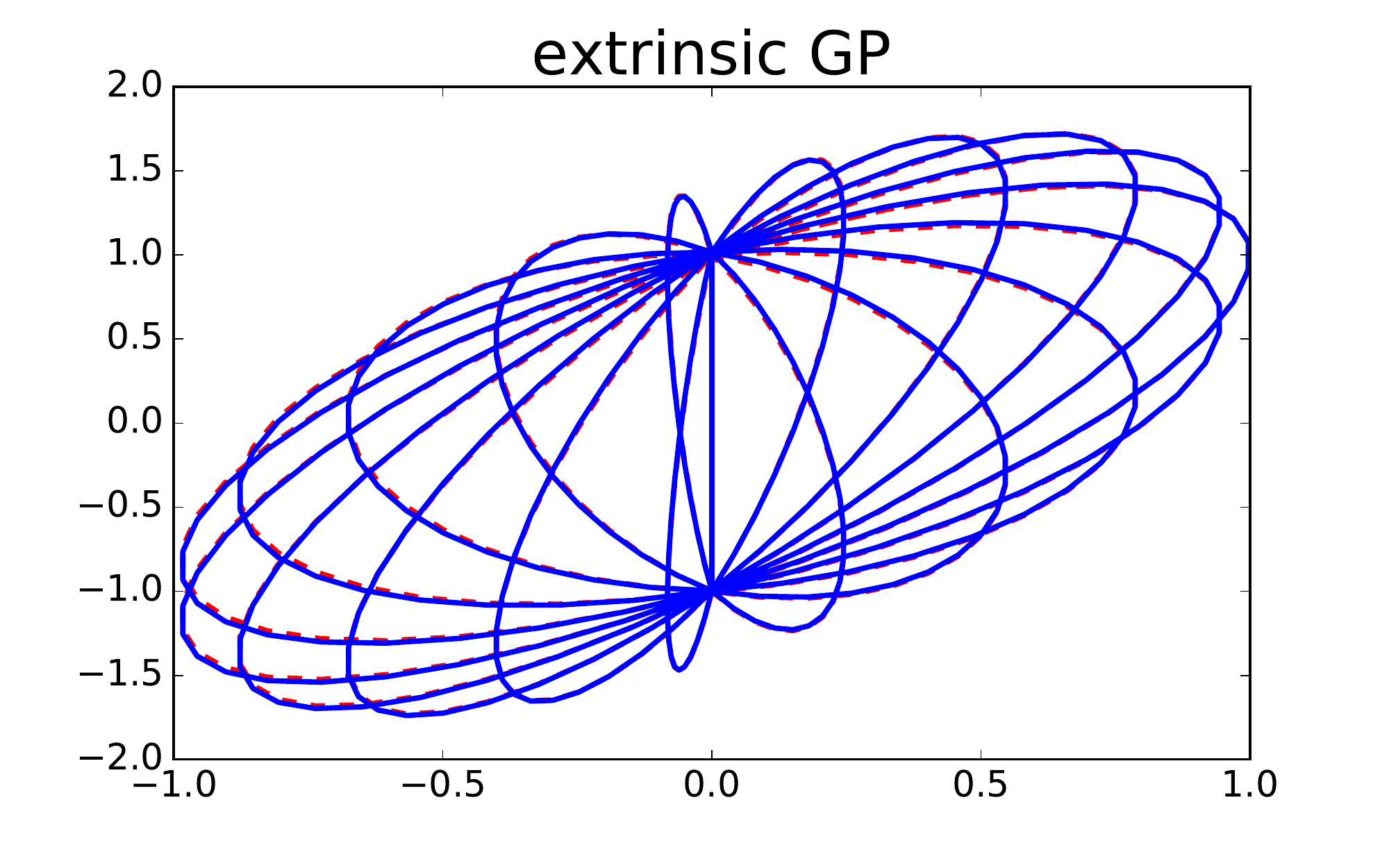}
  \includegraphics[width=0.4\linewidth]{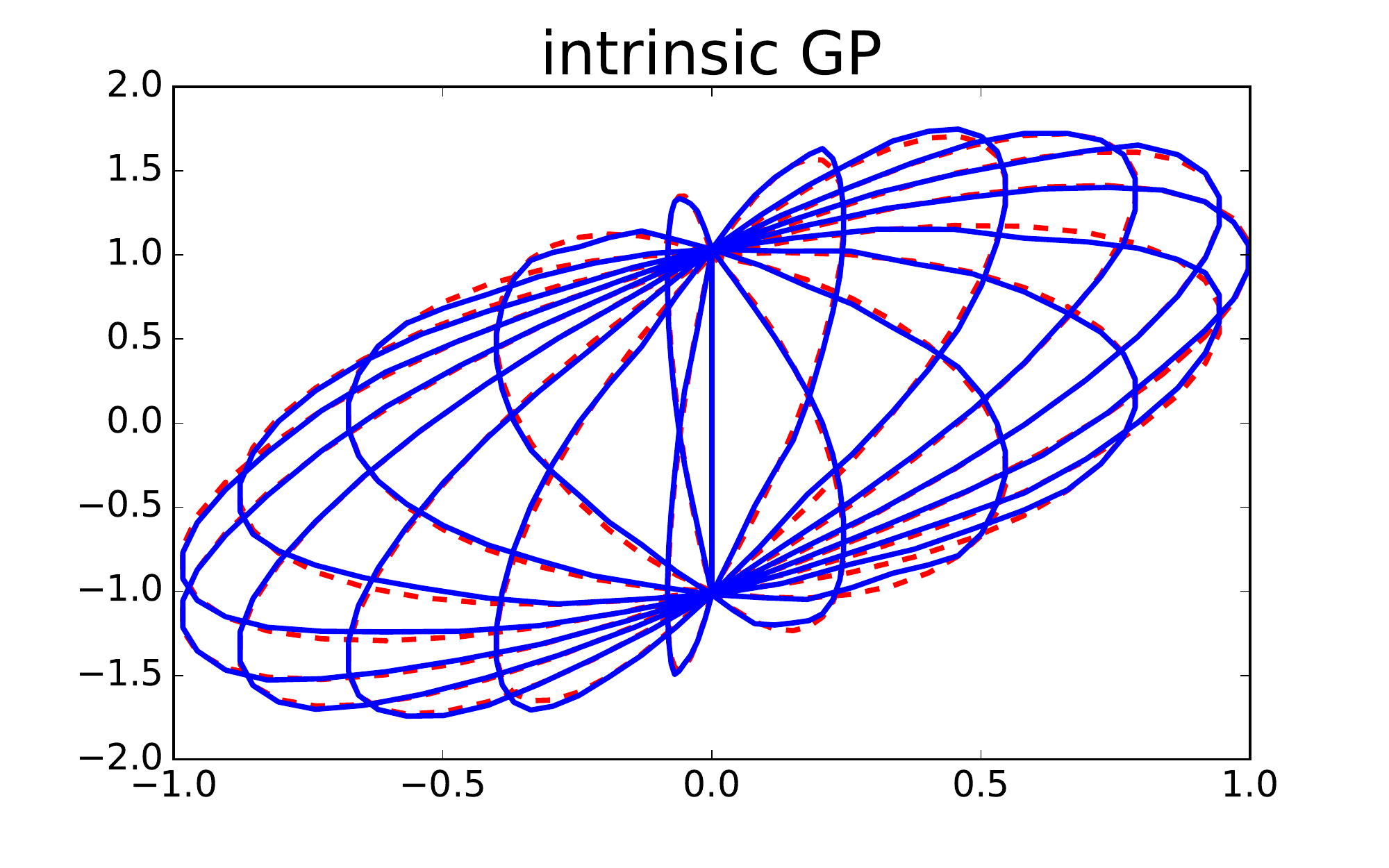}
     \caption{$\text{GP}$ predictive results using spherical exponential kernel vs eGP with an extrinsic kernel.  Truth is shown in red dashed lines and posterior mean estimates in blue.}
  \label{fig1}
\end{figure}



\subsection{Landmark-based shape spaces $\Sigma_2^k$}
We now apply eGP models to regression and classification on planar shapes. Planar shape spaces are one of the most important classes of landmark-based shape spaces with wide applications in biology and medical imaging. Such spaces were first studied in \cite{kendall77},  and in the pioneering work of \cite{books1} motivated by applications to biological shapes.


We first describe planar shapes. Let $z=(z_1,\ldots, z_k)$, with $z_1,\ldots, z_k\in \R^2$, be a set of $k$ landmarks. The planar shape $\Sigma_2^k$ is the collection of $z$s modulo the Euclidean motions including translation, scaling and rotation. One has $\Sigma_2^k=S^{2k-3}/SO(2)$,  the quotient of sphere by the action of $SO(2)$ (or modulo the effect of rotation),  the group of $2\times 2$ special orthogonal matrices;

A point in $\Sigma_2^k$ can be identified as the orbit of some $u\in S^{2k-3}$, which we denote as $\sigma(z)$. Viewing $z$ as a vector of complex numbers, one can embed $\Sigma_2^k$ into $S(k,\mathbb C)$, the space of $k\times k$ complex Hermitian matrices, via the Veronese-Whitney embedding (see e.g. \cite{rabimono}):
\begin{equation}
\label{eq-planaremb}
J(\sigma(z))=uu^*=((u_i\bar{u}_j))_{1\leq i,j\leq k}.
\end{equation}
One can verify that $J$ is equivariant (see \cite{kendall84}) with respect to the Lie group
$$H=SU(k)=\{A\in GL(k, \mathbb C), AA^*=I, \det(A)=I\},$$
with its action on $\Sigma_2^k$ induced by left multiplication.
This embedding $J$ will be used to construct covariance kernels for eGPs
on $\Sigma_2^k$.

As an example, we apply an eGP to a classification problem with predictors on $\Sigma_2^k$. We aim to classify the gorilla skull images from \cite{dimk}, which are represented as planar shapes with 8 landmarks, by gender. A binary GP  classification model is developed using 59 gorilla skull images. We take $y_i \in \{ 0,1 \}$, where $0$ represents a  female  and $1$  a male. 

We have the following model:
\begin{align}
y_i \sim Bernoulli(\pi_i), \;\;\;\pi_i=\Phi(F(x_i)),\;\; F(.) \sim \mbox{GP}(0,\mathcal{K}_{ext}),
\end{align}
where $\Phi$ is the standard normal cdf.


 Following \cite{williams1996} and \cite{neal2012}, we used Hamiltonian Monte Carlo (HMC) method for posterior computation. The likelihood is approximated using Laplace's  method as in \cite{williams1998}. Gamma priors  are used on the kernel hyperparameters, with Gamma(0.5,2) for the  length-scale and Gamma(50,1) for the magnitude paramter.    The number of MCMC iterations is 10,000 with a burn in of  3,000; The HMC estimates of  the kernel parameters are  shown in Figure \ref{spherical exponential}.

\begin{figure}[ht]
  \includegraphics[width=6cm]{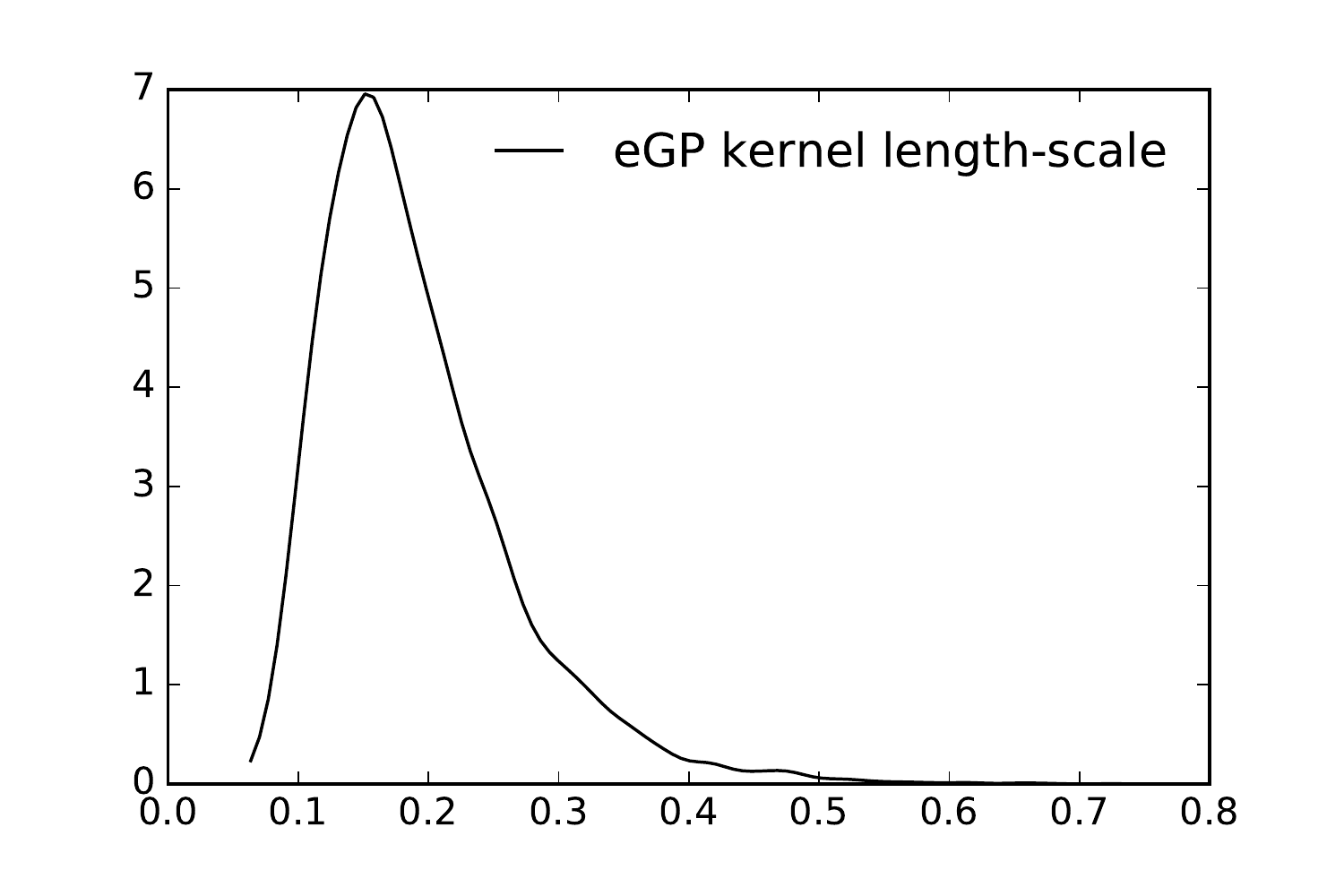}
    \includegraphics[width=6cm]{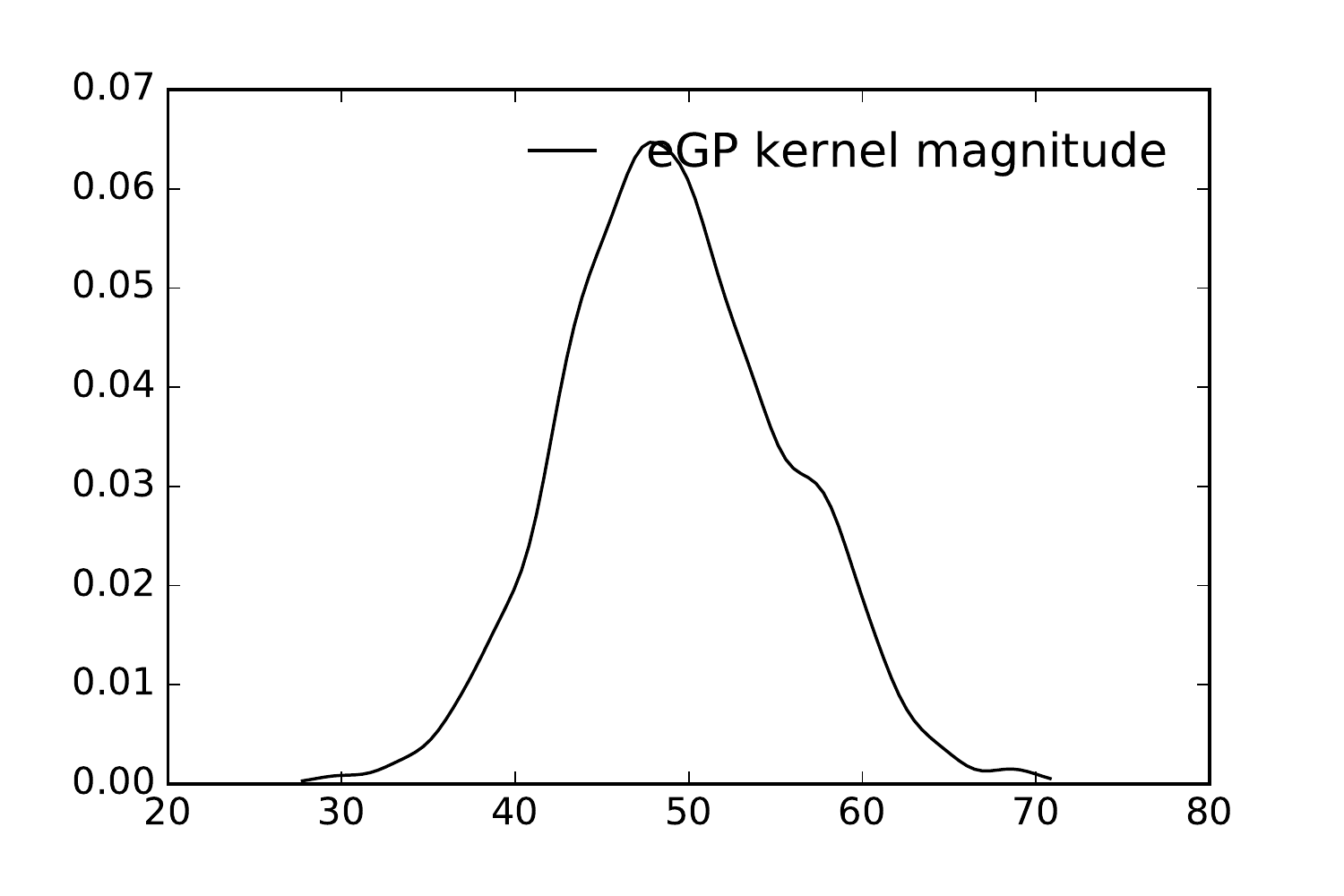}
  \caption{Posterior distributions of the eGP kernel parameters (the length-scale and  magnitude)}
  \label{spherical exponential}
\end{figure}


We use eight skull images  as testing data and all these images are successfully classified with our GP  classifier. The classification probabilities are provided in Table \ref{table:calcium}. The results are compared with a naive GP on the preshape data (modulo the effects of translation and scaling) without any embedding;  the latter completely failed at classification by  returning all the classification probabilities of 0.5. The results indicate that naive GPs are not suitable for complex manifolds not arising as submanifolds of an Euclidean space or  when simple representation of the space using Euclidean coordinates is not available.  In particular, for complex manifolds such as planar shapes, the naive representation of the data without properly incorporating the underlying geometry (e.g., via equivariant embeddings as in our case),  result in a posterior estimate of the latent function that is close to the prior mean (which is zero in our case) thus producing a classification probability of 0.5. 
\begin{table}[ht]
\caption{Planar shape classification  of gender based on gorilla skull shape.} 
\centering 
\begin{tabular}{c c c c c c c c c} 
\hline\hline 
Class& female & female & female  & female & male & male & male & male \\ [0.5ex] 
\hline 
GP  classification prob.& 7.2e-4 &  0.319 &0.029  &0.041 &0.96 &0.89 & 0.54 & 0.86 \\ [1ex] 
naive GP  classification prob.& 0.5 &  0.5 &0.5  &0.5 &0.5 &0.5 & 0.5 & 0.5 \\ [1ex]
\hline 
\end{tabular}
\label{table:calcium} 
\end{table}

\subsection{Diffusion tensor imaging and positive definite matrices}

Diffusion tensor imaging (DTI) is designed to measure the diffusion of water molecules in the brain; diffusion tends to be directional  along white matter tracks or fibers, corresponding to structural connections between brain regions along which substantial brain activity and communications occur. DTI data are now collected routinely in human studies, and there is abundant interest in using DTI to build better predictive models of cognitive traits and neuropsychiatric disorders.  The diffusion anisotropy characterized in terms of  \emph{diffusion matrices}, corresponding to $3\times 3$ positive definite matrices measured at each voxel in the brain.
We denote the space of all such matrices as $\SPD(3)$.

The space $\SPD(3)$ belongs to an important class of manifolds that possesses particular geometric structures, which should be taken into account in statistical analyses.
Our goal is to study the regression relationship between DTI-valued covariates and patient outcomes.


In order to carry out regression and classification on $SPD(3)$ using our eGP models, we need a nice embedding to construct the extrinsic kernels. There are a few natural embeddings of $\SPD(3)$ into Euclidean spaces. In particular, one can embed it into the space $\sym(3)$ of $3\times 3$ real symmetric matrices via the $\log$-map
\begin{align}
\label{eq-logmap}
\log: \SPD(3)\rightarrow \sym(3).
\end{align}
For $A\in\SPD(3)$ with a spectral decomposition (or diagonalization) $A=U\Lambda U^{-1}$, we have $\log(A)=U\log(\Lambda)U^{-1}$ where $\log(\Lambda)$ is the diagonal matrix whose diagonal entries are the logarithms of the diagonal entries of $\Lambda$. The embedding \eqref{eq-logmap} is in fact a diffeomorphism, and is equivariant with respect to the actions of $GL(3,\R)$, the $3\times 3$ general linear group, by conjugation. Indeed, for $h\in GL(3,\R)$, one has
\begin{align}
\log(hAh^{-1})=h\log(A)h^{-1}.
\end{align}
Given $A_1, A_2\in\SPD(3)$, their extrinsic distance under the embedding \eqref{eq-logmap} is given by
\begin{align}
\rho(A_1, A_2)=\|\log(A_1)-\log(A_2) \|,
\end{align}
where $\|\cdot \|$ denotes the Frobenius norm of matrices (i.e.~$\|A\|=\t{Tr}(AA^T)^{1/2}$). This extrinsic distance will be used to construct an eGP kernel in \eqref{eq-extrinsiker2}.

We now  consider a diffusion tensor imaging (DTI) data set  consisting of 46 subjects with 28 HIV+ subjects and 18 healthy controls.   Diffusion tensors were extracted along one atlas fiber tract of the splenium of the corpus callosum. The DTI data for all the subjects are  registered in the same atlas space based on arc lengths, with 75 tensors  obtained along the fiber tract of each subject. This data set has been studied in a regression setting in \cite{Yuan2012} and in the context of two sample testing (\cite{linclt}).  A  GP sampler is carried out between the  control group and the HIV+ group for each of the 75 sites  along the fiber tract. Therefore, 75  classifiers were run in total.   We aim to find out which sites of the splenium of the corpus callosum are most sensitive to influence by HIV.

14 subjects (six controls and eight HIV+) are used to test the HIV status classifiers ($0$ for healthy and $1$ for HIV+) using eGP models. A similar binary GP classification model is applied to the DTI data at each of  the prespecified 75 locations along the chosen tract. We have identified the top ten most sensitive sites indexed by the  arc length (location on the brain). The results are recorded in  Table \ref{table:grass}, which  shows the total number of correct GP  predictions of HIV status of the 14 tested subjects among the top ten sites.

\begin{table}[ht]
\caption{Diffusion tensor imaging results:  top 10 most sensitive sites to influence of HIV  } 
\centering 
\begin{tabular}{c c c c c c c c c c c} 
\hline\hline
arclength& $1.76$ & $4.42$ & $13.56$  & $26.52$ & $31.19$ & $33.16$  &$34.45$ & $35.62$ & $36.80$ & $37.11$ \\ [0.5ex] 
\hline
\# of correct GP  prediction & 11& 11 & 12 & 11 &11 &11&11 &11&12&11\\[1ex]
\hline
\end{tabular}
\label{table:grass}
\end{table}

Again HMC with Laplace approximation is used for model parameter inference. The posterior distribution of kernel hyperparameters for the GP classifier for one of the 75 sites along the fiber tract is shown in Figure \ref{PSDhmc}. Gamma(0.5,2) prior is used for kernel length-scale and Gamma(2.5,2) prior for kernel magnitude. The number of Monte Carlo iterations is 10,000 with a burn in of 3,000.  

\begin{figure}[ht]
  \includegraphics[width=6cm]{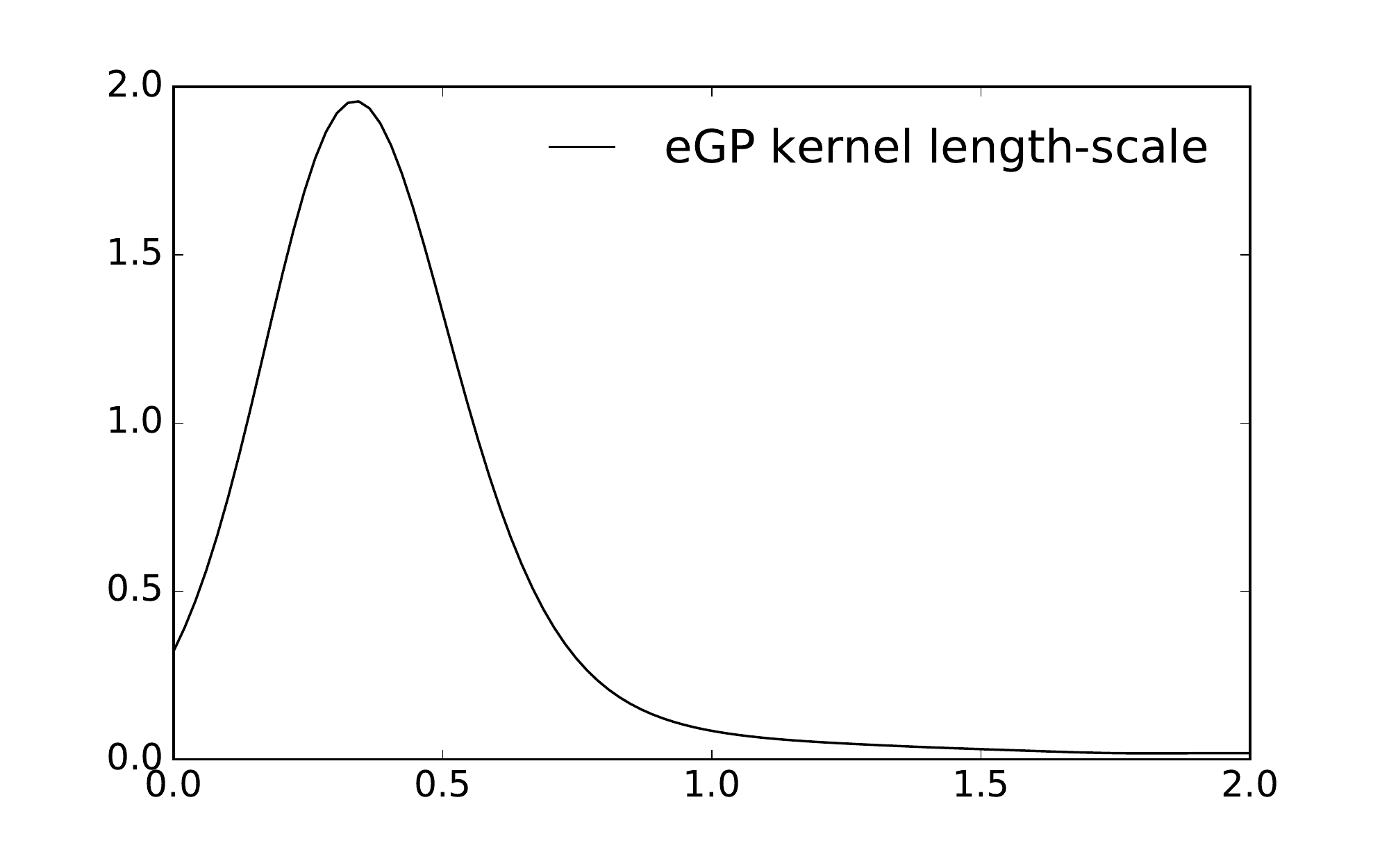}
    \includegraphics[width=6cm]{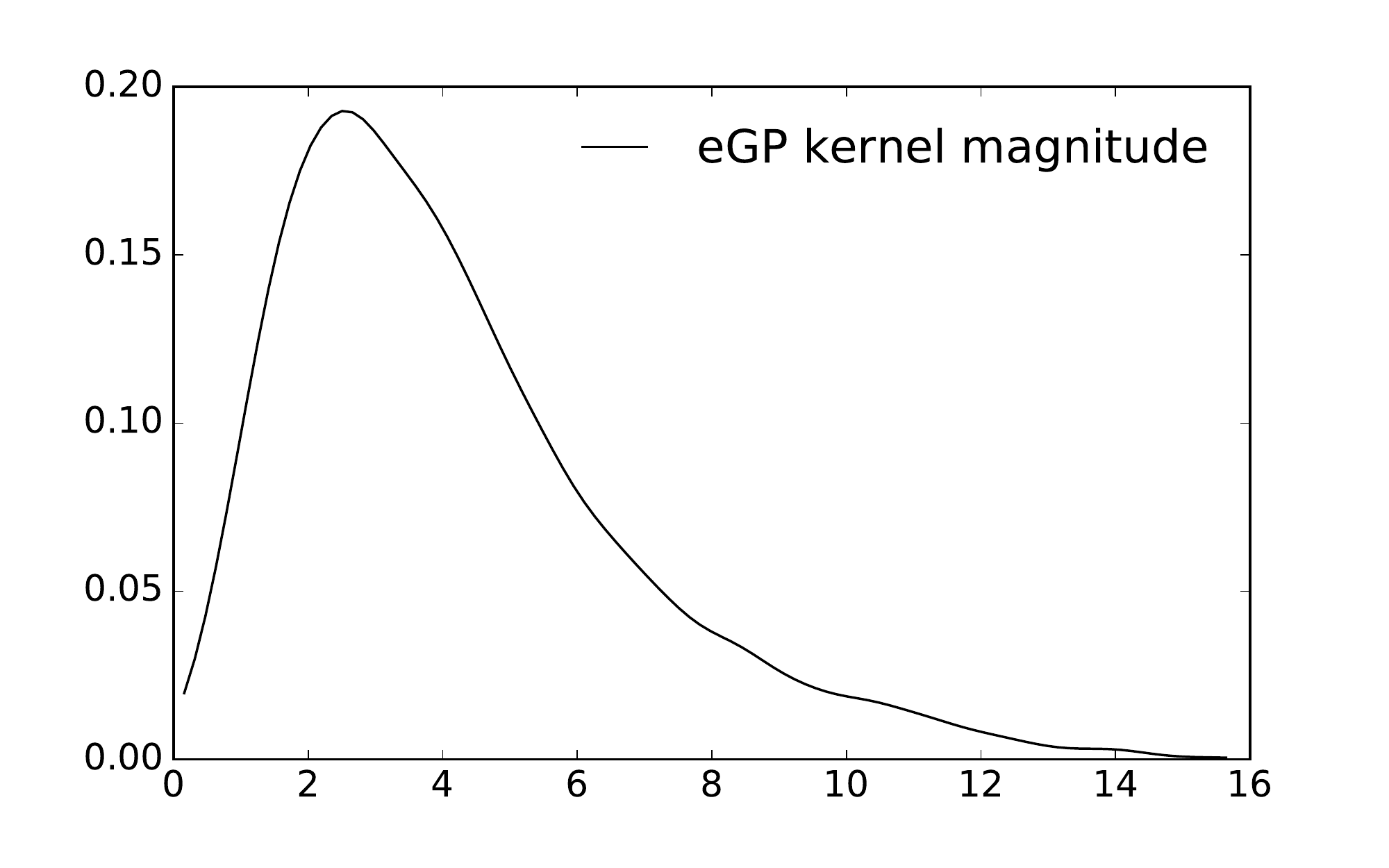}
  \caption{Posterior distribution for the eGP kernel covariance parameters in the diffusion tensor and HIV application.}
  \label{PSDhmc}
\end{figure}

\subsection{Stiefel manifolds and Grassmann manifolds (Grassmannians)}

We now consider regression and classification problems whose predictors lie on Stiefel or Grassmann manifolds. Given integers $m\geq k\geq 0$, the Stiefel manifold $V_k(\R^m)$ is the collection of all $k$-tuples of orthonormal vectors in $\R^m$, and the Grassmann manifold $Gr_k(\R^m)$ is the collection of all $k$-dimensional subspaces in $\R^m$.  Every $k$-tuple of orthonormal (hence linearly independent) vectors span a $k$-dimensional subspace, and every $k$-dimensional subspace is spanned by some $k$-tuple of orthonormal vectors. This means there is a surjective map $V_k(\R^m)\rightarrow Gr_k(\R^m)$. There is a natural action of $O(k)$, the group of $k\times k$ orthogonal matrices, on $V_k(\R^m)$ and any two $k$-tuples of orthonormal vectors span the same subspace precisely if they differ by an action of $O(k)$, which provides the identification $V_k(\R^m)/O(k)=Gr_k(\R^m)$. Grassmann manifolds have many applications in  signal processing and machine learning \citep{Kutyniok200964}.

There is an equivariant embedding of $Gr_k(\R^m)$ into a Euclidean space \citep{chikusebook}. Let $X\in V_k(\R^m)$ and $\sigma(X)=X\cdot O(k)$ be the $O(k)$-orbit of $X$ in $Gr_k(\R^m)=V_k(\mathbb R^m)/O(k)$. Note that
\begin{align*}
J(\sigma(X))=XX'
\end{align*}
defines an embedding $J$ of $Gr_k(\R^m)$ into the space of $m\times m$ matrices, which may be identified as $\R^{m^2}$. Also, it is equivariant with respect to the group $H=O(m)$ acting on $Gr_k(\R^m)$ via left multiplication on $\R^m$ and on $m\times m$ matrices by conjugation. Indeed, for $h\in H$, one has $J(h\sigma(X))=hXX'h'=\phi(h)J(\sigma(X))$, where $\phi(h)$ stands for conjugation by $h$.
Now the extrinsic distance between two points in $Gr_k(\R^m)$ is given by
$$\rho(\sigma(X_1),\sigma(X_2))=\| X_1X_1'-X_2X_2' \|,$$
where $\|\cdot\|$ is the Frobenius norm on matrices. We use the kernel \eqref{eq-extrinsiker2}.

\begin{remark}
The Stiefel manifold $V_k(\R^m)$ is naturally a submanifold of $\R^{m\times k}$ and the inclusion map is  an equivariant embedding.
\end{remark}

We now  apply the eGP model to data simulated from $ y= F(XX') + \epsilon$, 
where $X$ is an $m\times k$ matrix with $m=10$ the ambient dimension and $k=5$ the subspace dimension.
The data are simulated from the model with $F(X)=\beta XX' \beta$, where $\beta$ is some known vector.  We simulated 100, 200 and 300 training data points and additional 50 points for testing  with different signal-to-noise ratio levels. Table \ref{tab-grassmse} records the RMSE values.  As expected, the RMSE reduces with increasing training size and signal-to-noise ratio. 

\begin{table}[ht]
\caption{Simulation results for out-of-sample RMSE for prediction (for 50 testing points) based on predictors on the Grassmannian. } 
\centering 
\begin{tabular}{|c| c c c|  } 
\hline
\backslashbox{Training size}{Signal-to-noise ratio }& 10db & 20db &30db\\
\hline
$n=100$  & 1.25 & 0.6 &0.31\\ \hline
$n=200$  & 0.95& 0.31 & 0.098\\ \hline
$n=300$ & 0.77 &0.27 &0.089\\
\hline
\end{tabular}
\label{tab-grassmse}
\end{table}

%

The posterior distribution of kernel hyperparameters are estimated using HMC. A Gamma(2.5,2) prior is used for the kernel length-scale and Gamma(20,1) prior for the kernel magnitude.  The number of Monte Carlo iterations is 6000 with a burn in of 1000. 



\section{Properties of  eGPs }

In this section, we first study the properties of an eGP in terms of mean square differentiability. The smoothness of a stochastic process captures and quantifies the
intuition that inputs that are close (on a manifold) are likely to
produce similar output values.
Therefore, understanding the smoothness property is important for
interpolation and prediction.  In addition, we show that  (see Proposition \ref{prop-contract}) the posterior contraction rates of eGPs  are adaptive to the dimension of the underlying manifold instead of the ambient space where the manifolds are embedded onto building on results from \cite{yang2016}.
\subsection{Mean square differentiability}

We first give the definition of \emph{mean square differentiability and mean square derivative} of a stochastic process on a differentiable manifold.
Consider a smooth manifold $M$ and a stochastic process $w(x)$
indexed by $x\in M$.
Let $\mu(x)$ and $K(x_1,x_2)$ be the mean and covariance functions
of $w(x)$.

\begin{definition} $\phantom{a}$ (a) Let $x\in M$ and $v\in T_x M$.
Choose a smooth path $\gamma:(-\epsilon,\epsilon)\rightarrow M$
(for some $\epsilon>0$) such that $\gamma(0)=x$ and $\gamma'(0)=v$.
The stochastic process $w$ is {\bf mean squared (MS) differentiable}
\emph{at $x$ with respect to $v$} if, as $a\rightarrow 0$, the
random variable
\begin{align*}
\frac{w(\gamma(a))-w(x)}{a}
\end{align*}
converges to some limit $D_v w$ in mean squares, i.e.
\begin{align*}
\bb{E}\left[\left(\frac{w(\gamma(a))-w(x)}{a}-D_v w\right)^2\right]
\rightarrow 0.
\end{align*}
In this case, $D_v w$ is called the {\bf MS derivative}
\emph{of $w$ at $x$ with respect to $v$}.

(b) If $w$ is MS differentiable at $x$ with respect to every
tangent vector at that point, then we simply say that $w$ is
\emph{MS differentiable at $x$}.

(c) If $w$ is MS differentiable at every point in $M$, then we
simply say that $w$ is \emph{MS differentiable (in $M$)}.
In this case, for any tangent vector field $V$ in $M$, the random
variables $\{D_{V_x} w:x\in M\}$ constitute a stochastic process $D_V w$
in $M$, called the \emph{MS derivative of $w$ with respect to $V$}.
\end{definition}

\begin{remark}
The definition in (a) depends only on $x$ and $v$, but otherwise
not on the choice of $\gamma$.
This notion of MS differentiability generalizes the existing one
in Euclidean spaces.
\end{remark}

\begin{proposition} \label{MSDiff.cond}
If the mean function $\mu$ is differentiable at $x$ and the
covariance function $K$ is of class $C^2$ at $(x,x)$, then
the stochastic process $w$ is MS differentiable at $x$.
\end{proposition}


\begin{proof}
Since $\mu$ is differentiable at $x$, the statement will hold
for $w$ if it also holds for $w-\mu$, whose mean function is $0$.
Hence we may assume $\mu=0$ without loss of generality.

Suppose $\gamma:(-\epsilon,\epsilon)\rightarrow M$ is a smooth
path with $\gamma(0)=x$ (for some $\epsilon>0$).
Let
\begin{align*}
v=\gamma'(0)\in T_x M,\qquad
v^{(1)}=(v,0),\; v^{(2)}=(0,v)\;\in\;
  T_{(x,x)}(M\times M)=T_x M\times T_x M\,.
\end{align*}
For $a\in(-\epsilon,0)\cup(0,\epsilon)$, consider the random
variable
\begin{align*}
D_a = \frac{w(\gamma(a))-w(x)}{a}\,.
\end{align*}
It suffices to show that $D_a$ has a limit in mean squares
(i.e.~in $L^2$) as $a\rightarrow 0$.
Notice that
\begin{align*}
\bb{E}[D_a D_b]
=\frac{1}{ab}\Big(
  K(\gamma(a),\gamma(b))-K(\gamma(a),x)-K(x,\gamma(b))+K(x,x)\Big)
\end{align*}
Since $K$ is of class $C^2$ at $(x,x)$, as
$(a,b)\rightarrow(0,0)$, we have
\begin{align*}
\bb{E}[D_a D_b]\rightarrow\big(D_{v^{(1)}}D_{v^{(2)}}K\big)(x,x).
\end{align*}
It follows that, under the same limit,
\begin{align*}
\bb{E}[(D_a-D_b)^2]
&=\bb{E}[D_a^2]+\bb{E}[D_b^2]-2\bb{E}[D_a D_b] \\
&\rightarrow
\big(D_{v^{(1)}}D_{v^{(2)}}K\big)(x,x)
+\big(D_{v^{(1)}}D_{v^{(2)}}K\big)(x,x)
-2\big(D_{v^{(1)}}D_{v^{(2)}}K\big)(x,x)
=0
\end{align*}
Therefore, as $a\rightarrow 0$, $D_a$ satisfies the Cauchy
condition with respect to the $L^2$ norm and, by completeness,
admits an $L^2$ limit.
\end{proof}

\begin{proposition} \label{MSDer.mean.cov}
If the mean function $\mu$ is differentiable in $M$ and the
covariance function $K$ is of class $C^2$ in $M\times M$,
then the stochastic process $w$ is MS differentiable in $M$.
In this case, for any tangent vector field $V$ in $M$, the
MS derivative $D_V w$ has mean function $D_V \mu$ and covariance
function $D_{V^{(1)}}D_{V^{(2)}}K$, where $V^{(1)}$ and $V^{(2)}$
are  the tangent vector fields in $M\times M$ with
$V^{(1)}_{(x_1,x_2)}=(V_{x_1},0)$ and
$V^{(2)}_{(x_1,x_2)}=(0,V_{x_2})$.
\end{proposition}

\begin{proof}
The first statement is immediate from Proposition \ref{MSDiff.cond}.
For $i=1,2$, let $x_i\in M$ and
$\gamma_i:(-\epsilon,\epsilon)\rightarrow M$ be a smooth path with
$\gamma_i(0)=x_i$ and $\gamma'_i(0)=V_{x_i}$.
By the Cauchy-Schwarz inequality and the MS differentiability
of $w$, we have
\begin{align*}
&\bb{E}\left[\left(
  (D_V w)(x_1)-\frac{w(\gamma_1(a))-w(x_1)}{a}
  \right)\right]
  \rightarrow 0,\quad
  \t{as }a\rightarrow 0 \\
\iff\quad
&\bb{E}[(D_V w)(x_1)]-\frac{\mu(\gamma_1(a))-\mu(x_1)}{a}\rightarrow 0,
\quad\t{as }a\rightarrow 0
\end{align*}
so that $\bb{E}[(D_V w)(x_1)]=(D_V\mu)(x_1)$.
Now let $\tilde w=w-\mu$.
Similarly as above, we have
\begin{align*}
&\bb{E}\left[\left(
  (D_V \tilde w)(x_1)-\frac{\tilde w(\gamma_1(a))-\tilde w(x_1)}{a}
  \right) \tilde w(x_2) \right]
  \rightarrow 0,\quad
  \t{as }a\rightarrow 0 \\
\iff\quad
&\bb{E}[(D_V \tilde w)(x_1)\, \tilde w(x_2)]
  -\frac{K(\gamma_1(a),x_2)-K(x_1,x_2)}{a}\rightarrow 0,
\quad\t{as }a\rightarrow 0
\end{align*}
so that $\bb{E}[(D_V \tilde w)(x_1)\, \tilde w(x_2)]=(D_{V^{(1)}}K)(x_1,x_2)$.
Similarly again, we also have
\begin{align*}
&\bb{E}\left[
  \left((D_V \tilde w)(x_1)-\frac{\tilde w(\gamma_1(a))-\tilde w(x_1)}{a}\right)
  \left((D_V \tilde w)(x_2)-\frac{\tilde w(\gamma_2(b))-\tilde w(x_2)}{b}\right)
  \right]\rightarrow 0 \\
\iff\quad
&\bb{E}[(D_V \tilde w)(x_1)\,(D_V \tilde w)(x_2)]
  -\frac{K(\gamma_1(a),\gamma_2(b))
    -K(\gamma_1(a),x_2)-K(x_1,\gamma_2(b))+K(x_1,x_2)}{ab}
  \rightarrow 0
\end{align*}
as $(a,b)\rightarrow(0,0)$, which means
\begin{align*}
\bb{E}&[(D_V \tilde w)(x_1)\,(D_V \tilde w)(x_2)]  \\
&=(D_{V^{(2)}}D_{V^{(1)}}K)(x_1,x_2)
+(D_{V^{(1)}}D_{V^{(2)}}K)(x_1,x_2)
-(D_{V^{(1)}}D_{V^{(2)}}K)(x_1,x_2) \\
&=(D_{V^{(1)}}D_{V^{(2)}}K)(x_1,x_2).
\end{align*}
This completes the proof.
\end{proof}

\begin{corollary} \label{nMSDiff.cond}
If $\mu$ is of class $C^n$ and $K$ is of class $C^{2n}$,
then $w$ is $n$-times MS differentiable.
\end{corollary}

\begin{proof}
Repeatedly apply Proposition \ref{MSDer.mean.cov}.
\end{proof}

\begin{example}
Suppose $J:M\rightarrow\R^D$ is an embedding of $M$ into
a (higher-dimensional) Euclidean space $\R^D$.
Given a stochastic process $w$ in $\R^D$, we can pull it
back to a stochastic process $J^*w$ in $M$, with
\begin{align*}
(J^*w)(x)=w(J(x)),\quad\t{for }x\in M.
\end{align*}
Clearly, if the mean and covariance functions of $w$ are $\mu$
and $K$, then the mean and covariance functions of $J^*f$
are $J^*\mu$ and $(J\times J)^*K$.
Also, if $\mu$ is $C^n$, $K$ is $C^{2n}$ and $J$ is $C^{2n}$ as well,
then $J^*\mu$ is $C^n$ and $(J\times J)^*K$ is $C^{2n}$;
and hence by Corollary \ref{nMSDiff.cond}, $J^*w$ is $n$-times MS
differentiable.

For example, if $w$ is a Gaussian process in $\R^D$ with
a Mat\'ern-$\nu$ covariance function (and zero mean), then
$J^*w$ is an $\lfloor\frac{\nu-1}{2}\rfloor$-times MS
differentiable Gaussian process in $M$;
and if $w$ is a Gaussian process in $\R^D$ with
a squared-exponential covariance function, then $J^*w$ is
an infinitely MS differentiable Gaussian process in $M$.
\end{example}
%

\subsection{Posterior contraction rates of eGPs}

In this short subsection, we explore the posterior contraction rates of a regression model on a manifold with eGP  as the prior for the regression function.  Posterior contraction rates measure how fast the posterior concentrates in small neighborhoods of the true regression function, providing frequentist asymptotic guarantees on the behavior of the eGP posterior.  Given data $(x_i,y_i)$ with $x_i\in M$ and $y_i\in \R$ ($i=1,\ldots, n$),  assume the  regression model \eqref{eq-model1} where $y_i=F(x_i)+\epsilon_i$,  $x_i\in M$ and $\epsilon_i\sim N(0,\sigma^2)$. The prior distribution   $\Pi (F)$ will be given by the eGP with the covariance kernel \eqref{eq-extrinsiker2} (with a fixed magnitude). The length-scale parameter $\beta$ is assumed a prior  $\pi_{\beta}$  such that $\beta^d$ follows a gamma distribution $\text{Gamma}(a_0, b_0)$, where $d$ is the dimension of manifold.
For simplicity in exposition, assume $\sigma$ is known though the results are straightforward to generalize to unknown $\sigma$.  The posterior distribution of $F$ is then given by
\begin{align}
\Pi\left(U\mid (x_1,y_1),\ldots, (x_n,y_n)\right)=\dfrac{\int_{U} \prod_{i=1}^n N(y_i;F(x_i),\sigma^2)\Pi(dF)}{\int \prod_{i=1}^n N(y_i;F(x_i),\sigma^2)\Pi(dF)}
\end{align} 
where $U$ is a measurable set in the space of regression functions.  Let $F_0$ be the true regression function. We say the eGP posterior \emph{contracts to $F_0$ at a rate of $\epsilon_n$} if
\begin{align}
\Pi\left(U_{\epsilon_n}(F_0)^C\mid (x_1,y_1),\ldots, (x_n,y_n)\right)\rightarrow 0, \; a.s. P_{F_0}^n, 
\end{align} 
where $U_{\epsilon_n}(F_0)^C=\{F: d_{ \mathcal{M}}(F,F_0)> C\epsilon_n\}$, as $n\rightarrow \infty$ for some large constant $C$ and distance  $d_{ \mathcal{M}}$.
We have the following proposition.
\begin{proposition}
\label{prop-contract}
Assume the regression model \eqref{eq-model1} with an eGP prior with covariance kernel \eqref{eq-extrinsiker2}, the following holds.
\begin{itemize}
\item[(a)] Assume $M$ is a smooth manifold and the covariates are from a fixed design.  Let $F_0\in C^s(M)$ ($s\leq 2$), the $s$-H\"older smooth class of functions on $M$,  then the posterior distribution of eGP contracts to the true regression function  $F_0$ at a rate of $\epsilon_n=n^{-s/(2s+d)}(\log n)^{d+1}$ with $d_{ \mathcal{M}}(F, F_0)=\frac{1}{n}\sum_{i=1}^n |F(x_i)-F_0(x_i)|.$

\item[(b)] Assume $M$ is a smooth manifold and the covariates are from a random design with $x_i\sim g(\cdot)$, $i=1,\ldots,n$, for some distribution $g(\cdot)$ on $M$. Then the results in part (a) hold with 
$U_{\epsilon_n}(F_0)^C=\{F: \int_{x\in M} (F_A(x)-F_0(x))^2g(dx)<\epsilon_n\}$, where $F_A(x)=(f\vee (-A))\wedge A$, for some $A$ large enough.
\end{itemize}

\end{proposition}

\begin{proof}

(a) Given the embedding $J: M\rightarrow \R^D$, $\tilde{M}=J(M)$ is a $d$-dimensional submanifold of $\R^D$. Any  function $F\in \mathcal{M}$ on $M$ induces a function $\tilde F=F\circ J^{-1}$ on $\tilde M$.  One has
\begin{align*}
y_i=\tilde F(\tilde x_i)+\epsilon_i,
\end{align*}
where $\tilde x_i=J(x_i)\in\tilde{M}$. Then by Theorem 2.1 of \cite{yang2016}, one has
\begin{align*}
\Pi\left(\tilde U_{\epsilon_n}(\tilde F_0)^C\mid (\tilde x_1,y_1),\ldots, (\tilde x_n,y_n)\right)\rightarrow 0
\end{align*}
where $\tilde U_{\epsilon_n}(F_0)=\{\tilde F : \frac{1}{n}\sum_{i=1}^n\mid  \tilde F(\tilde x_i)-\tilde F_0(\tilde x_i)\mid<\epsilon_n\}.$ There is a one-to-one correspondence (a bijection) between $\tilde F$ and $F$, and one has
$ U_{\epsilon_n}(F_0)=\{F: \frac{1}{n}\sum_{i=1}^n |F(x_i)-F_0(x_i)|=\frac{1}{n}\sum_{i=1}^n\mid  \tilde F(\tilde x_i)-\tilde F_0(\tilde x_i)\mid<\epsilon_n\}.$ Then
\begin{align*}
\Pi\left(U_{\epsilon_n}(F_0)^C\mid ( x_1,y_1),\ldots, ( x_n,y_n)\right)\rightarrow 0,
\end{align*}
where $\epsilon_n$ is given in part (a). 

(b) Similar proofs follow from part (a) noting that there is one-to-one correspondence between
$\{\tilde F: \int_{\tilde M}(\tilde F(\tilde x)-\tilde F_0(\tilde x))^2\tilde g(dx)<\epsilon_n\}$ and 
$\{F: \int_M (F(x)-F_0(x))^2g(x)dx<\epsilon_n\}$, where $\tilde g(x)$ is the density on $\tilde M$ induced by the embedding $J$ and the density $g(x)$ on $M$. 
\end{proof}

\section{Discussion and conclusion}

We propose a general extrinsic framework for constructing Gaussian processes  on manifolds for regression and classification with manifold-valued predictors.  Such models are general, easy to implement and shown to inherit good properties from Gaussian processes on Euclidean spaces.  Applications are considered by applying eGP models to regression and classification problems with predictors on a large class of manifolds ranging from spheres, landmark-based shapes spaces, to the spaces of positive definite matrices and Grassmannians. Our work will likely help practitioners  make more accurate predictions or diagnoses  based on medical imaging. Although the work focuses on regression and classification, the eGPs can be used in much broader settings  such as in exponential family models for the response $y_i$ given $x_i$,  which allows Poisson regression etc.  In addition, eGPs can be certainly used for spatial modeling where the spatial space is some geometric space such as the sphere and other geometric spaces.  Future work will be devoted to constructing applicable covariance kernels employing the intrinsic Riemannian geometry of manifolds, which are only available now for a very limited class of manifolds, and also constructing valid GP models for spaces  beyond manifolds such as  stratified spaces of interests.


\section*{Acknowledgment}

We  thank Professor Hongtu Zhu for providing us the diffusion tensor imaging data used  in Section 3. The contribution of LL is funded by  NSF grants IIS1663870 and Career 1654579.

\bibliographystyle{biom}
\bibliography{refs-LL}

\end{document}